\documentclass[runningheads]{llncs}
\usepackage{graphicx}

\usepackage[hidelinks]{hyperref}

\usepackage[utf8]{inputenc}
\usepackage{microtype} 

\usepackage{stix}
\usepackage{mathtools}
\usepackage[centertableaux]{ytableau}
\usepackage{tikz}
\usetikzlibrary{arrows,positioning}
\usepackage[capitalize,noabbrev]{cleveref}

\newcommand{\defin}[1]{\emph{#1}}

\newcommand{\lltri}{\ensuremath{\lltriangle}}
\newcommand{\ultri}{\ensuremath{\ultriangle}}
\newcommand{\urtri}{\ensuremath{\urtriangle}}
\newcommand{\lrtri}{\ensuremath{\lrtriangle}}
\newcommand{\laser}{\ensuremath{\conictaper}}
\newcommand{\solid}{\ensuremath{\blacksquare}}
\newcommand{\movable}{\ensuremath{\boxdot}}
\newcommand{\goal}{\ensuremath{\bigoplus}}

\newcommand{\inputleft}{\ensuremath{\rightdotarrow}}
\newcommand{\inputtop}{\rotatebox[origin=c]{270}{\ensuremath{\rightdotarrow}}}

\newcommand{\laserleft}{\ensuremath{\rightdotarrow*}}
\newcommand{\lasertop}{\rotatebox[origin=c]{270}{\ensuremath{\rightdotarrow*}}}

\newcommand{\outright}{\ensuremath{\rightarrow}}
\newcommand{\outdown}{\ensuremath{\downarrow}}

\newcommand{\ANDcolor}{green}
\newcommand{\ORcolor}{violet!50}
\newcommand{\VARIABLEcolor}{yellow}
\newcommand{\SWITCHcolor}{cyan}

\newcommand{\ORcolorBack}{violet!25}
\newcommand{\VARIABLEcolorBack}{yellow!50}
\newcommand{\SWITCHcolorBack}{cyan!30}
\newcommand{\TRIANGcolorBack}{red!45}

\author{Per Alexandersson\inst{1}\orcidID{0000-0003-2176-0554} \and
Petter Restadh\inst{1}\orcidID{0000-0002-3411-8766}}
\authorrunning{P. Alexandersson \and P. Restadh}
\institute{KTH The Royal Insitute of Technology, 100 44 Stockholm, Sweden \\
\email{per.w.alexandersson@gmail.com}\\
\email{petterre@kth.se}
}

\begin{document}

\title{LaserTank is NP-complete}

\maketitle

\begin{abstract}
We show that the classical game \emph{LaserTank} is $\mathrm{NP}$-complete,
even when the tank movement is restricted to a single column 
and the only blocks appearing on the board are mirrors and solid blocks.
We show this by reducing $3$-SAT instances to LaserTank puzzles.

\end{abstract}

\section{Introduction}

From Wikipedia: ``\emph{LaserTank (also known as Laser Tank) is a computer puzzle game requiring 
logical thinking to solve a variety of levels}''. It was first released on the Windows platform in 1995, 
and a similar game was released in 1998 for the graphing calculator Texas Instruments Ti-83, 
under the name \emph{Laser Mayhem}\footnote{\url{https://www.ticalc.org/archives/files/fileinfo/95/9532.html}}.
To our knowledge, the complexity of LaserTank has not been studied before,
while several other classical games have been shown to be NP-complete, NP-hard or PSPACE-complete.
For example,
\emph{Sokoban} \cite{DorZwick1999},
\emph{Tetris} \cite{DemaineHohenbergerLibenNowell2002},
\emph{Rush Hour} \cite{FlakeBaum2002}, and
\emph{Minesweeper} \cite{Kaye2000} to list a few.

In this short note, we prove the following.
\begin{theorem}\label{thm:MainTheorem}
LaserTank is $\mathrm{NP}$-complete.
\end{theorem}

It should be noted that one can perhaps apply more general meta-theoretical 
approaches for puzzle games and planning games in particular,
to prove NP-completeness. It is likely that the framework by G. Viglietta \cite{Viglietta2013} ---
which can be applied to games such as \emph{Boulder Dash}, \emph{Pipe Mania} and \emph{Starcraft} ---
can successfully be applied to LaserTank as well.
We opted for a self-contained hands-on approach where 3-SAT is reduced to LaserTank.
Furthermore, we only use a small subset of the available pieces in the original game,
as well as restrict the movement of the tank in two directions. 
These restrictions have the benefit that they imply that the \emph{Laser Mayhem} variant is also NP-complete.

\subsection{Short background on 3-SAT}

A \defin{$3$-SAT expression} $E$ is a conjunction of clauses, where each clause involves exactly three distinct literals.
A literal is either a boolean variable, or its negation.
The \defin{$3$-SAT problem} states: Determine if $E$ is \defin{satisfiable} ---  that is, there is an assignment of 
truth values to the variables that makes $E$ true.
For example, 
$
E = (x\vee y \vee \neg z) \wedge (\neg x\vee \neg y \vee \neg z) \wedge (x\vee \neg y \vee z)
$
is such a conjunction, and the assignment $x,z=\mathtt{true}$, $y=\mathtt{false}$ shows that $E$
is satisfiable. Determining satisfiability of a $3$-SAT expression is an NP-complete problem \cite{Cook1971}.

\section{LaserTank}

LaserTank is a turn-based single-player puzzle game played on a $2$-dimensional grid (the \emph{board}),
where in each turn, the player either moves the tank, or fires a laser from the tank.
The laser interacts with different \emph{pieces} on the board, and the goal is to
hit a certain piece with the laser.
The pieces\footnote{For a complete list of pieces available in the official game, see \url{https://lasertanksolutions.blogspot.com/p/in-my-opinionlaser-tank-is-best-logic.html}}
we use are
\emph{mirrors}  $\{\lltri, \, \urtri, \, \lrtri, \, \ultri\}$,
\emph{solid blocks} $\solid$,
\emph{movable blocks} $\movable$,
the \emph{tank} $\laser$, and
the \emph{goal} $\goal$.
The tank is the only piece directly controlled by the player, and the laser exits the tank from the front, 
which is the pointy end of $\laser$.
In our version, the tank is restricted to sideways movement only, see \cref{ex:tankEx}
The tank can fire a laser from the front. If the laser hits a mirror on a slanted edge it is reflected.
When a mirror is hit on one of the two (non-reflective) short edges by the laser, 
the mirror is pushed in the direction of the laser.
A movable block is pushed one step if it is hit by the laser.
A movable block or a mirror is only pushed if the tile directly behind it is empty.
The aim of the puzzle is to hit the goal piece with the laser.
The solid blocks do not allow lasers or the tank to pass through and they do not move when hit by the laser.
The following example shows all game mechanics in action.

\begin{example}\label{ex:tankEx}
 Here is a small instance of the problem, with a step-by-step solution.
 The tank fires a laser which moves a mirror (1),
 then takes one step sideways, (2). It then shoots a laser at the movable block (3),
 and finally moves in position to have a clear shot of the goal (4).
 \[
 \ytableausetup{boxsize=1.0em}
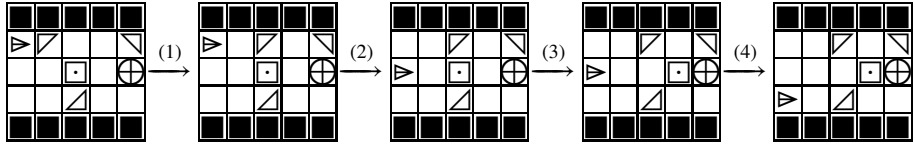
\begin{ytableau}
\solid &\solid & \solid & \solid & \solid \\
 \laser & \ultri & &  & \urtri \\
 &  & \movable &    & \goal\\
 &  &\lrtri & & \\
\solid &  \solid & \solid & \solid & \solid  \\
\end{ytableau}
\xrightarrow{(1)}
\begin{ytableau}
\solid &\solid & \solid & \solid & \solid \\
 \laser &  & \ultri &  & \urtri \\
 &  & \movable &    & \goal\\
 &  &\lrtri & & \\
\solid &  \solid & \solid & \solid & \solid  \\
\end{ytableau}
\xrightarrow{(2)}
\begin{ytableau}
\solid &\solid & \solid & \solid & \solid \\
  &  & \ultri &  & \urtri \\
 \laser &  & \movable &    & \goal\\
 &  &\lrtri & & \\
\solid &  \solid & \solid & \solid & \solid  \\
\end{ytableau}
\xrightarrow{(3)}
\begin{ytableau}
\solid &\solid & \solid & \solid & \solid \\
 &  & \ultri &  & \urtri \\
\laser &  &  & \movable   & \goal\\
 &  &\lrtri & & \\
\solid &  \solid & \solid & \solid & \solid  \\
\end{ytableau}
\xrightarrow{(4)}
\begin{ytableau}
\solid &\solid & \solid & \solid & \solid \\
 &  & \ultri &  & \urtri \\
 &  &  & \movable   & \goal\\
\laser &  &\lrtri & & \\
\solid &  \solid & \solid & \solid & \solid  \\
\end{ytableau}
\]
\end{example}

Our goal is now to construct puzzles which imitates an instance of $3$-SAT.
We employ so called \defin{gadgets} that emulate boolean operations.
Below, we let $\inputleft$, $\inputtop$ indicate the inputs to the gadgets (considered as boolean variables),
and $\{\laserleft,\,\lasertop\}$ indicate inputs that are always available as clear shots from
the tank. The latter are used for producing the output of the gadget.

\textbf{The \texttt{and} gadget}.
The configuration in \cref{fig:gadgets}a serves as our \texttt{and}-gadget.
We need to shoot through both $A$ and $B$ in order to allow 
for $A\wedge B = X$ as output.
Notice that the two movable blocks can only be moved up, right and down. 
If we want the gadget to produce an output through $X$, all movable blocks must be moved out of the way.
This can only be accomplished if the movable block must have been moved to 
the right via activation from both $A$ and $B$, which shows that the gadget is indeed an \texttt{and}-gadget.
The \texttt{and}-gadget can easily be generalized to more than two inputs.

\begin{figure}[!ht]
\[
 \ytableausetup{boxsize=1.2em}
\substack{
\begin{ytableau}
\none&\solid & \solid & \solid \\
\none&\solid & \ultri & \outright & \none[X] \\
\none&\solid &  & \solid \\
\none[B]&\inputleft & \movable & & \solid \\
\none&\solid &  & \solid \\
\none[A]&\inputleft & \movable & & \solid \\
\none&\solid &  & \solid \\
\none&\laserleft & \lrtri & \solid \\
\none&\solid &  \solid  & \solid \\
\end{ytableau}
\\ \, \\ \text{(a) AND}}
\qquad\quad
\substack{
\begin{ytableau}
\none & \none[A] &\none & \none[B]&\none & \none[C] \\
\solid & \inputtop &\solid& \inputtop &\solid& \inputtop & \solid  \\
\solid & \ultri &\solid& \ultri &\solid& \ultri  & \solid \\
\solid &  &&&  & & \outright & \none[X]\\
\laserleft &&&  &  & \lrtri  & \solid \\
\laserleft &  && \lrtri &&   & \solid \\
\laserleft & \lrtri &\solid  & \solid  &\solid&\solid& \solid \\
\end{ytableau}
\\ \, \\ \text{(b) THREE-OR}
}
\qquad
\substack{
\begin{ytableau}
\none & \solid &\solid &\solid&\solid\\
\none[X]&\laserleft&\urtri & \solid&\solid \\
 \none&\solid&&\solid&\ultri&\outright&\none[\neg X]\\
\none[\neg X]&\laserleft&\movable&&\lrtri&\solid\\
\none&\solid&&\solid &\solid\\
\none&\solid &\lltri&&\outright&\none[X] \\
\none&\solid&\solid&\solid&\solid  \\
\end{ytableau}
\\ \, \\ \text{(c) LITERAL}
}
\quad
\substack{
\begin{ytableau}
\none&&\lasertop\\
\none&\solid&\urtri\\
\none[X] &\inputleft&\\
\none&\ultri&\lrtri\\
\none&\outdown&
\end{ytableau}
\\ \, \\  \text{\hspace{5mm}(d) SWITCH}
}
\]
\caption{The and-gadget, three-or-gadget, literal-gadget, and switch-gadget.}\label{fig:gadgets}
\end{figure}

\textbf{The \texttt{three-or} gadget} is depicted in \cref{fig:gadgets}b.
If either of the inputs $A$, $B$ or $C$ are available, then $X$ allows for output. 
The only way to produce output from $X$ is to move a $\ultri$ to the same row as $X$. 
The $\ultri$ can only be moved into that row from above and thus we must have some input from $A$, $B$ or $C$
in order for a laser to pass out through $X$. 
Thus the \texttt{three-or} gadget works the way intended.

\textbf{The \texttt{literal} gadget} is depicted in \cref{fig:gadgets}c.
This gadget emulates a literal, with two different mutually exclusive outputs 
depending on the choice of value of the literal.
To unlock $X$ as output, fire once through $\neg X$ first.
This moves the movable block out of the way but 
prevents $\neg X$ from being available as output. 
Similarly for $\neg X$.

\textbf{The \texttt{switch} gadget} is depicted in \cref{fig:gadgets}d.
The \texttt{switch}-gadget is our main building block 
for encoding an instance of a $3$-SAT problem.
It allows for the input $X$ to be available first as output to the right, then 
redirected down. This allows $X$ to be used in multiple \texttt{or}-clauses.

\begin{example}
In the puzzle in \cref{fig:switchExample}, only a single ``input'', $X$, is available.
However, with the \texttt{switches} we can redirect input $X$ to activate the \texttt{and} gadget.
Notice the two $\urtri$ pieces that are required to activate the \texttt{switches}
and that the rightmost switch gadget must be used first in order to solve the puzzle.
This is also true in the general setup, where switches should be used from right to left.
\begin{figure}[!ht]
\begin{tikzpicture}[scale=0.75]
\node at (0,0) {$
\ytableausetup{boxsize=1.2em}
\begin{ytableau}
\solid&\solid&\solid&\solid&\solid&\solid&\solid&\solid&\solid&\solid&\solid\\
\laser&&&&&&&&\urtri&\solid&\solid\\
&&&&&\urtri&\solid&&&&\solid\\
&\solid&\solid&\solid&&&&&&&\solid\\
&\solid&\solid&\solid&*(\SWITCHcolor)\solid&*(\SWITCHcolor)\urtri&&*(\SWITCHcolor)\solid&*(\SWITCHcolor)\urtri&&\solid\\
&X&\laserleft&&*(\SWITCHcolor)&*(\SWITCHcolor)&&*(\SWITCHcolor)&*(\SWITCHcolor)&&\solid\\
&\solid&\solid&\solid&*(\SWITCHcolor)\ultri&*(\SWITCHcolor)\lrtri&&*(\SWITCHcolor)\ultri&*(\SWITCHcolor)\lrtri&&\solid\\
&&\urtri&\solid&&&&&&&\solid\\
&*(\ANDcolor)\solid&*(\ANDcolor)&*(\ANDcolor)\solid&*(\ANDcolor)&*(\ANDcolor)\solid&*(\ANDcolor)\solid&*(\ANDcolor)&*(\ANDcolor)\solid&*(\ANDcolor)\solid&*(\ANDcolor)\solid\\
&*(\ANDcolor)\solid&*(\ANDcolor)\lltri&*(\ANDcolor)&*(\ANDcolor)\movable&*(\ANDcolor)&*(\ANDcolor)&*(\ANDcolor)\movable&*(\ANDcolor)&*(\ANDcolor)\urtri&*(\ANDcolor)\solid\\
&*(\ANDcolor)\solid&*(\ANDcolor)\solid&*(\ANDcolor)\solid&*(\ANDcolor)&*(\ANDcolor)\solid&*(\ANDcolor)\solid&*(\ANDcolor)&*(\ANDcolor)\solid&*(\ANDcolor)&*(\ANDcolor)\solid\\
&&&&*(\ANDcolor)\solid&&&*(\ANDcolor)\solid&\solid&\goal&\solid\\
\solid&\solid&\solid&\solid&\solid&\solid&\solid&\solid&\solid&\solid&\solid\\
\end{ytableau}
$};
\end{tikzpicture}
\;
\begin{tikzpicture}[rotate=90, every node/.style={rotate=90}]
\begin{scope}[scale=0.57]
\path[fill = \TRIANGcolorBack, draw=black] (0,2) rectangle (5.5,0);
\path[fill = \SWITCHcolorBack, draw=black] (0,0) rectangle (5.5,-7);
\path[fill = \ORcolorBack, draw=black] (0,-7) rectangle (5.5,-10.5);
\path[fill = \VARIABLEcolorBack, draw=black] (0,0) rectangle (-1,-7);
\draw[step=0.5cm,gray,very thin] (-1.5,2) grid (7.5,-11);
\path[fill = \VARIABLEcolor, draw=black] (0,0) rectangle (-1,-2);
\node at (-0.5, -1) {$X_1$};
\path[fill = \VARIABLEcolor, draw=black] (0,-2) rectangle (-1,-4);
\node at (-0.5, -3) {$X_2$};
\node at (-0.5,-4.5) {$\vdots$};
\path[fill = \VARIABLEcolor, draw=black] (0,-5) rectangle (-1,-7);
\node at (-0.5, -6) {$X_n$};
\path[fill = \ORcolor, draw=black] (0,-7) rectangle (1.5,-8);
\node at (0.75, -7.5) {$\scriptstyle{\mathtt{OR}}$};
\path[fill = \ORcolor, draw=black] (1.5,-8) rectangle (3,-9);
\node at (2.25, -8.5) {$\scriptstyle{\mathtt{OR}}$};
\node at (3.5,-9.25) {$\ddots$};
\path[fill = \ORcolor, draw=black] (4,-9.5) rectangle (5.5,-10.5);
\node at (4.75, -10) {$\scriptstyle{\mathtt{OR}}$};
\path[fill = \ANDcolor, draw=black] (5.5,-11) rectangle (6.5,-7);
\node at (6, -9) {$\scriptstyle{\mathtt{AND}}$};
\path[fill = \SWITCHcolor, draw=black] (0,0) rectangle (0.5,-1);
\path[fill = \SWITCHcolor, draw=black] (0.5,-3) rectangle (1,-4);
\path[fill = \SWITCHcolor, draw=black] (1,-5) rectangle (1.5,-6);
\path[fill = \SWITCHcolor, draw=black] (1.5,-1) rectangle (2,-2);
\path[fill = \SWITCHcolor, draw=black] (2,-3) rectangle (2.5,-4);
\path[fill = \SWITCHcolor, draw=black] (2.5,-4) rectangle (3,-5);

\node at (-1.25,-0.25){$\laser$};
\node at (0.25,.25){$\urtri$};
\node at (0.75,.5){$\urtri$};
\node at (1.25,.75){$\urtri$};
\node at (1.75,1){$\urtri$};
\node at (2.25,1.25){$\urtri$};
\node at (2.75,1.5){$\urtri$};
\node at (7, -7.5) {$\goal$};
\end{scope}
\end{tikzpicture}
\caption{
Left: A small puzzle showing the use of two \texttt{switch}-gadgets and one \texttt{and}-gadget.
Right: Layout of a general $3$-SAT puzzle.
Above each \texttt{or}-gadget are three switches, corresponding to the three
literals involved on the or-clause.}
\label{fig:switchExample}\label{fig:generalLayout}
\end{figure}
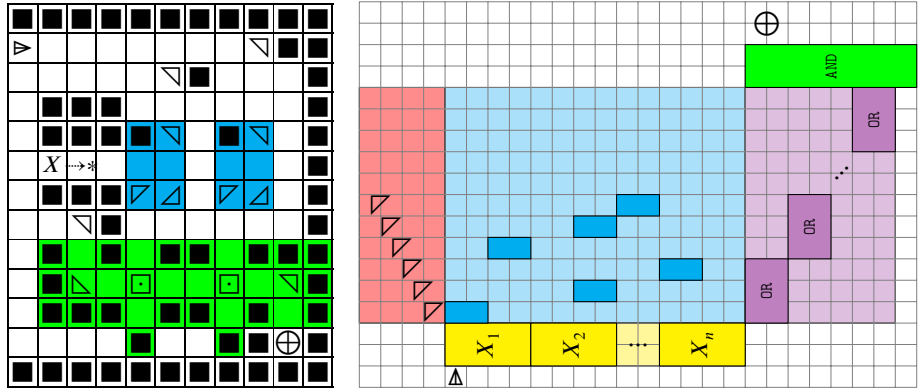
\end{example}

\subsection{The reduction}

A 3-SAT expression may now be encoded as a LaserTank puzzle as follows.
There is one \texttt{literal}-gadget for each variable appearing in the expression,
a \texttt{three-or-gadget} for each \texttt{or}-clause, and a single \texttt{and}-gadget 
with multiple inputs is used to bind all together.
The puzzle is designed such that the output of the \texttt{and}-gadget is the only way to hit the goal.
The general layout of such a puzzle is shown in \cref{fig:generalLayout}.
For each three-or-clause in the $3$-SAT expression,
three switches are placed on the board corresponding to the three literals involved.
In other words, the clauses of the $3$-SAT expression are encoded via \texttt{switch}-gadgets.
The switches can always be activated via the $\urtri$ pieces at the 
top of the board as in \cref{fig:switchExample}.
As a concrete example, the expression $(A\vee B\vee \neg C)\wedge (A\vee \neg B\vee C)$ 
is encoded as the puzzle shown in \cref{fig:fullExample}.
\begin{figure}
\begin{tikzpicture}[scale=0.7, rotate=90, every node/.style={transform shape}]
\node at (0,0) {$
\begin{ytableau}[]
\none&\solid&\solid&\solid&\solid&\solid&\solid&\solid&\solid&\solid&\solid&\solid&\solid&\solid&\solid&\solid&\solid&\solid&\solid&\solid&\solid&\solid&\solid&\solid&\solid\\
\none&\laser&&&&&&*(\TRIANGcolorBack)&*(\TRIANGcolorBack)&*(\TRIANGcolorBack)&*(\TRIANGcolorBack)&*(\TRIANGcolorBack)&*(\TRIANGcolorBack)&*(\TRIANGcolorBack)&*(\TRIANGcolorBack)&*(\TRIANGcolorBack)&*(\TRIANGcolorBack)&*(\TRIANGcolorBack)&*(\TRIANGcolorBack)&*(\TRIANGcolorBack)\urtri&&&&&\\
\none&&&&&&&*(\TRIANGcolorBack)&*(\TRIANGcolorBack)&*(\TRIANGcolorBack)&*(\TRIANGcolorBack)&*(\TRIANGcolorBack)&*(\TRIANGcolorBack)&*(\TRIANGcolorBack)&*(\TRIANGcolorBack)&*(\TRIANGcolorBack)&*(\TRIANGcolorBack)&*(\TRIANGcolorBack)\urtri&*(\TRIANGcolorBack)&*(\TRIANGcolorBack)&&&&&\\
\none&&&&&&&*(\TRIANGcolorBack)&*(\TRIANGcolorBack)&*(\TRIANGcolorBack)&*(\TRIANGcolorBack)&*(\TRIANGcolorBack)&*(\TRIANGcolorBack)&*(\TRIANGcolorBack)&*(\TRIANGcolorBack)&*(\TRIANGcolorBack)\urtri&*(\TRIANGcolorBack)&*(\TRIANGcolorBack)&*(\TRIANGcolorBack)&*(\TRIANGcolorBack)&&&&&\\
\none&&&&&&&*(\TRIANGcolorBack)&*(\TRIANGcolorBack)&*(\TRIANGcolorBack)&*(\TRIANGcolorBack)&*(\TRIANGcolorBack)&*(\TRIANGcolorBack)\urtri&*(\TRIANGcolorBack)&*(\TRIANGcolorBack)&*(\TRIANGcolorBack)&*(\TRIANGcolorBack)&*(\TRIANGcolorBack)&*(\TRIANGcolorBack)&*(\TRIANGcolorBack)&&&&&\\
\none&&&&&&&*(\TRIANGcolorBack)&*(\TRIANGcolorBack)&*(\TRIANGcolorBack)&*(\TRIANGcolorBack)\urtri&*(\TRIANGcolorBack)&*(\TRIANGcolorBack)&*(\TRIANGcolorBack)&*(\TRIANGcolorBack)&*(\TRIANGcolorBack)&*(\TRIANGcolorBack)&*(\TRIANGcolorBack)&*(\TRIANGcolorBack)&*(\TRIANGcolorBack)&&&&&\\
\none&&&&&&&*(\TRIANGcolorBack)&*(\TRIANGcolorBack)\urtri&*(\TRIANGcolorBack)&*(\TRIANGcolorBack)&*(\TRIANGcolorBack)&*(\TRIANGcolorBack)&*(\TRIANGcolorBack)&*(\TRIANGcolorBack)&*(\TRIANGcolorBack)&*(\TRIANGcolorBack)&*(\TRIANGcolorBack)&*(\TRIANGcolorBack)&*(\TRIANGcolorBack)&&&&&\\
\none&&*(\VARIABLEcolor)\solid&*(\VARIABLEcolor)\solid&*(\VARIABLEcolor)\solid&*(\VARIABLEcolor)\solid&&*(\SWITCHcolorBack)&*(\SWITCHcolorBack)&*(\SWITCHcolorBack)&*(\SWITCHcolorBack)&*(\SWITCHcolorBack)&*(\SWITCHcolorBack)&*(\SWITCHcolorBack)&*(\SWITCHcolorBack)&*(\SWITCHcolorBack)&*(\SWITCHcolorBack)&*(\SWITCHcolorBack)&*(\SWITCHcolorBack)&*(\SWITCHcolorBack)&&&&&\\
\none[A]&&*(\VARIABLEcolor)\inputleft&*(\VARIABLEcolor)\urtri&*(\VARIABLEcolor)\solid&*(\VARIABLEcolor)\solid&&*(\SWITCHcolorBack)&*(\SWITCHcolorBack)&*(\SWITCHcolorBack)&*(\SWITCHcolorBack)&*(\SWITCHcolorBack)&*(\SWITCHcolorBack)&*(\SWITCHcolorBack)&*(\SWITCHcolorBack)&*(\SWITCHcolorBack)&*(\SWITCHcolorBack)&*(\SWITCHcolorBack)&*(\SWITCHcolorBack)&*(\SWITCHcolorBack)&&&&&\\
\none&&*(\VARIABLEcolor)\solid&*(\VARIABLEcolor)&*(\VARIABLEcolor)\solid&*(\VARIABLEcolor)\ultri&*(\VARIABLEcolor)&*(\SWITCHcolorBack)&*(\SWITCHcolorBack)&*(\SWITCHcolorBack)&*(\SWITCHcolorBack)&*(\SWITCHcolorBack)&*(\SWITCHcolorBack)&*(\SWITCHcolorBack)&*(\SWITCHcolorBack)&*(\SWITCHcolorBack)&*(\SWITCHcolorBack)&*(\SWITCHcolorBack)&*(\SWITCHcolorBack)&*(\SWITCHcolorBack)&&&&&\\
\none[\neg A]&&*(\VARIABLEcolor)\inputleft&*(\VARIABLEcolor)\movable&*(\VARIABLEcolor)&*(\VARIABLEcolor)\lrtri&*(\VARIABLEcolor)\solid&*(\SWITCHcolorBack)&*(\SWITCHcolorBack)&*(\SWITCHcolorBack)&*(\SWITCHcolorBack)&*(\SWITCHcolorBack)&*(\SWITCHcolorBack)&*(\SWITCHcolorBack)&*(\SWITCHcolorBack)&*(\SWITCHcolorBack)&*(\SWITCHcolorBack)&*(\SWITCHcolorBack)&*(\SWITCHcolorBack)&*(\SWITCHcolorBack)&&&&&\\
\none&&*(\VARIABLEcolor)\solid&*(\VARIABLEcolor)&*(\VARIABLEcolor)\solid&*(\VARIABLEcolor)\solid&&*(\SWITCHcolor)\solid&*(\SWITCHcolor)\urtri&*(\SWITCHcolorBack)&*(\SWITCHcolorBack)&*(\SWITCHcolorBack)&*(\SWITCHcolorBack)&*(\SWITCHcolorBack)&*(\SWITCHcolor)\solid&*(\SWITCHcolor)\urtri&*(\SWITCHcolorBack)&*(\SWITCHcolorBack)&*(\SWITCHcolorBack)&*(\SWITCHcolorBack)&&&&&\\
\none&&*(\VARIABLEcolor)\solid&*(\VARIABLEcolor)\lltri&*(\VARIABLEcolor)&*(\VARIABLEcolor)&&*(\SWITCHcolor)&*(\SWITCHcolor)&*(\SWITCHcolorBack)&*(\SWITCHcolorBack)&*(\SWITCHcolorBack)&*(\SWITCHcolorBack)&*(\SWITCHcolorBack)&*(\SWITCHcolor)&*(\SWITCHcolor)&*(\SWITCHcolorBack)&*(\SWITCHcolorBack)&*(\SWITCHcolorBack)&*(\SWITCHcolorBack)&&&&&\\
\none&&*(\VARIABLEcolor)\solid&*(\VARIABLEcolor)\solid&*(\VARIABLEcolor)\solid&*(\VARIABLEcolor)\solid&&*(\SWITCHcolor)\ultri&*(\SWITCHcolor)\lrtri&*(\SWITCHcolorBack)&*(\SWITCHcolorBack)&*(\SWITCHcolorBack)&*(\SWITCHcolorBack)&*(\SWITCHcolorBack)&*(\SWITCHcolor)\ultri&*(\SWITCHcolor)\lrtri&*(\SWITCHcolorBack)&*(\SWITCHcolorBack)&*(\SWITCHcolorBack)&*(\SWITCHcolorBack)&&&&&\\
\none&&*(\VARIABLEcolor)\solid&*(\VARIABLEcolor)\solid&*(\VARIABLEcolor)\solid&*(\VARIABLEcolor)\solid&&*(\SWITCHcolorBack)&*(\SWITCHcolorBack)&*(\SWITCHcolorBack)&*(\SWITCHcolorBack)&*(\SWITCHcolorBack)&*(\SWITCHcolorBack)&*(\SWITCHcolorBack)&*(\SWITCHcolorBack)&*(\SWITCHcolorBack)&*(\SWITCHcolorBack)&*(\SWITCHcolorBack)&*(\SWITCHcolorBack)&*(\SWITCHcolorBack)&&&&&\\
\none[B]&&*(\VARIABLEcolor)\inputleft&*(\VARIABLEcolor)\urtri&*(\VARIABLEcolor)\solid&*(\VARIABLEcolor)\solid&&*(\SWITCHcolorBack)&*(\SWITCHcolorBack)&*(\SWITCHcolorBack)&*(\SWITCHcolorBack)&*(\SWITCHcolorBack)&*(\SWITCHcolorBack)&*(\SWITCHcolorBack)&*(\SWITCHcolorBack)&*(\SWITCHcolorBack)&*(\SWITCHcolor)\solid&*(\SWITCHcolor)\urtri&*(\SWITCHcolorBack)&*(\SWITCHcolorBack)&&&&&\\
\none&&*(\VARIABLEcolor)\solid&*(\VARIABLEcolor)&*(\VARIABLEcolor)\solid&*(\VARIABLEcolor)\ultri&*(\VARIABLEcolor)&*(\SWITCHcolorBack)&*(\SWITCHcolorBack)&*(\SWITCHcolorBack)&*(\SWITCHcolorBack)&*(\SWITCHcolorBack)&*(\SWITCHcolorBack)&*(\SWITCHcolorBack)&*(\SWITCHcolorBack)&*(\SWITCHcolorBack)&*(\SWITCHcolor)&*(\SWITCHcolor)&*(\SWITCHcolorBack)&*(\SWITCHcolorBack)&&&&&\\
\none[\neg B]&&*(\VARIABLEcolor)\inputleft&*(\VARIABLEcolor)\movable&*(\VARIABLEcolor)&*(\VARIABLEcolor)\lrtri&*(\VARIABLEcolor)\solid&*(\SWITCHcolorBack)&*(\SWITCHcolorBack)&*(\SWITCHcolorBack)&*(\SWITCHcolorBack)&*(\SWITCHcolorBack)&*(\SWITCHcolorBack)&*(\SWITCHcolorBack)&*(\SWITCHcolorBack)&*(\SWITCHcolorBack)&*(\SWITCHcolor)\ultri&*(\SWITCHcolor)\lrtri&*(\SWITCHcolorBack)&*(\SWITCHcolorBack)&&&&&\\
\none&&*(\VARIABLEcolor)\solid&*(\VARIABLEcolor)&*(\VARIABLEcolor)\solid&*(\VARIABLEcolor)\solid&&*(\SWITCHcolorBack)&*(\SWITCHcolorBack)&*(\SWITCHcolor)\solid&*(\SWITCHcolor)\urtri&*(\SWITCHcolorBack)&*(\SWITCHcolorBack)&*(\SWITCHcolorBack)&*(\SWITCHcolorBack)&*(\SWITCHcolorBack)&*(\SWITCHcolorBack)&*(\SWITCHcolorBack)&*(\SWITCHcolorBack)&*(\SWITCHcolorBack)&&&&&\\
\none&&*(\VARIABLEcolor)\solid&*(\VARIABLEcolor)\lltri&*(\VARIABLEcolor)&*(\VARIABLEcolor)&&*(\SWITCHcolorBack)&*(\SWITCHcolorBack)&*(\SWITCHcolor)&*(\SWITCHcolor)&*(\SWITCHcolorBack)&*(\SWITCHcolorBack)&*(\SWITCHcolorBack)&*(\SWITCHcolorBack)&*(\SWITCHcolorBack)&*(\SWITCHcolorBack)&*(\SWITCHcolorBack)&*(\SWITCHcolorBack)&*(\SWITCHcolorBack)&&&&&\\
\none&&*(\VARIABLEcolor)\solid&*(\VARIABLEcolor)\solid&*(\VARIABLEcolor)\solid&*(\VARIABLEcolor)\solid&&*(\SWITCHcolorBack)&*(\SWITCHcolorBack)&*(\SWITCHcolor)\ultri&*(\SWITCHcolor)\lrtri&*(\SWITCHcolorBack)&*(\SWITCHcolorBack)&*(\SWITCHcolorBack)&*(\SWITCHcolorBack)&*(\SWITCHcolorBack)&*(\SWITCHcolorBack)&*(\SWITCHcolorBack)&*(\SWITCHcolorBack)&*(\SWITCHcolorBack)&&&&&\\
\none&&*(\VARIABLEcolor)\solid&*(\VARIABLEcolor)\solid&*(\VARIABLEcolor)\solid&*(\VARIABLEcolor)\solid&&*(\SWITCHcolorBack)&*(\SWITCHcolorBack)&*(\SWITCHcolorBack)&*(\SWITCHcolorBack)&*(\SWITCHcolorBack)&*(\SWITCHcolorBack)&*(\SWITCHcolorBack)&*(\SWITCHcolorBack)&*(\SWITCHcolorBack)&*(\SWITCHcolorBack)&*(\SWITCHcolorBack)&*(\SWITCHcolorBack)&*(\SWITCHcolorBack)&&&&&\\
\none[C]&&*(\VARIABLEcolor)\inputleft&*(\VARIABLEcolor)\urtri&*(\VARIABLEcolor)\solid&*(\VARIABLEcolor)\solid&&*(\SWITCHcolorBack)&*(\SWITCHcolorBack)&*(\SWITCHcolorBack)&*(\SWITCHcolorBack)&*(\SWITCHcolor)\solid&*(\SWITCHcolor)\urtri&*(\SWITCHcolorBack)&*(\SWITCHcolorBack)&*(\SWITCHcolorBack)&*(\SWITCHcolorBack)&*(\SWITCHcolorBack)&*(\SWITCHcolorBack)&*(\SWITCHcolorBack)&&&&&\\
\none&&*(\VARIABLEcolor)\solid&*(\VARIABLEcolor)&*(\VARIABLEcolor)\solid&*(\VARIABLEcolor)\ultri&*(\VARIABLEcolor)&*(\SWITCHcolorBack)&*(\SWITCHcolorBack)&*(\SWITCHcolorBack)&*(\SWITCHcolorBack)&*(\SWITCHcolor)&*(\SWITCHcolor)&*(\SWITCHcolorBack)&*(\SWITCHcolorBack)&*(\SWITCHcolorBack)&*(\SWITCHcolorBack)&*(\SWITCHcolorBack)&*(\SWITCHcolorBack)&*(\SWITCHcolorBack)&&&&&\\
\none[\neg C]&&*(\VARIABLEcolor)\inputleft&*(\VARIABLEcolor)\movable&*(\VARIABLEcolor)&*(\VARIABLEcolor)\lrtri&*(\VARIABLEcolor)\solid&*(\SWITCHcolorBack)&*(\SWITCHcolorBack)&*(\SWITCHcolorBack)&*(\SWITCHcolorBack)&*(\SWITCHcolor)\ultri&*(\SWITCHcolor)\lrtri&*(\SWITCHcolorBack)&*(\SWITCHcolorBack)&*(\SWITCHcolorBack)&*(\SWITCHcolorBack)&*(\SWITCHcolorBack)&*(\SWITCHcolorBack)&*(\SWITCHcolorBack)&&&&&\\
\none&&*(\VARIABLEcolor)\solid&*(\VARIABLEcolor)&*(\VARIABLEcolor)\solid&*(\VARIABLEcolor)\solid&&*(\SWITCHcolorBack)&*(\SWITCHcolorBack)&*(\SWITCHcolorBack)&*(\SWITCHcolorBack)&*(\SWITCHcolorBack)&*(\SWITCHcolorBack)&*(\SWITCHcolorBack)&*(\SWITCHcolorBack)&*(\SWITCHcolorBack)&*(\SWITCHcolorBack)&*(\SWITCHcolorBack)&*(\SWITCHcolor)\solid&*(\SWITCHcolor)\urtri&&&&&\\
\none&&*(\VARIABLEcolor)\solid&*(\VARIABLEcolor)\lltri&*(\VARIABLEcolor)&*(\VARIABLEcolor)&&*(\SWITCHcolorBack)&*(\SWITCHcolorBack)&*(\SWITCHcolorBack)&*(\SWITCHcolorBack)&*(\SWITCHcolorBack)&*(\SWITCHcolorBack)&*(\SWITCHcolorBack)&*(\SWITCHcolorBack)&*(\SWITCHcolorBack)&*(\SWITCHcolorBack)&*(\SWITCHcolorBack)&*(\SWITCHcolor)&*(\SWITCHcolor)&&&&&\\
\none&&*(\VARIABLEcolor)\solid&*(\VARIABLEcolor)\solid&*(\VARIABLEcolor)\solid&*(\VARIABLEcolor)\solid&&*(\SWITCHcolorBack)&*(\SWITCHcolorBack)&*(\SWITCHcolorBack)&*(\SWITCHcolorBack)&*(\SWITCHcolorBack)&*(\SWITCHcolorBack)&*(\SWITCHcolorBack)&*(\SWITCHcolorBack)&*(\SWITCHcolorBack)&*(\SWITCHcolorBack)&*(\SWITCHcolorBack)&*(\SWITCHcolor)\ultri&*(\SWITCHcolor)\lrtri&*(\ANDcolor)\solid&*(\ANDcolor)\solid&*(\ANDcolor)\solid&\solid&\solid\\
\none&&&&&&*(\ORcolor)\solid&*(\ORcolor)&*(\ORcolor)\solid&*(\ORcolor)&*(\ORcolor)\solid&*(\ORcolor)&*(\ORcolor)\solid&*(\ORcolorBack)&*(\ORcolorBack)&*(\ORcolorBack)&*(\ORcolorBack)&*(\ORcolorBack)&*(\ORcolorBack)&*(\ORcolorBack)&*(\ANDcolor)\solid&*(\ANDcolor)\ultri&*(\ANDcolor)&\goal&\solid\\
\none&&&&&&*(\ORcolor)\solid&*(\ORcolor)\ultri&*(\ORcolor)\solid&*(\ORcolor)\ultri&*(\ORcolor)\solid&*(\ORcolor)\ultri&*(\ORcolor)\solid&*(\ORcolorBack)&*(\ORcolorBack)&*(\ORcolorBack)&*(\ORcolorBack)&*(\ORcolorBack)&*(\ORcolorBack)&*(\ORcolorBack)&*(\ANDcolor)\solid&*(\ANDcolor)&*(\ANDcolor)\solid&\solid&\solid\\
\none&&&&&&*(\ORcolor)\solid&*(\ORcolor)&*(\ORcolor)&*(\ORcolor)&*(\ORcolor)&*(\ORcolor)&*(\ORcolor)&*(\ORcolorBack)&*(\ORcolorBack)&*(\ORcolorBack)&*(\ORcolorBack)&*(\ORcolorBack)&*(\ORcolorBack)&*(\ORcolorBack)&*(\ANDcolor)&*(\ANDcolor)\movable&*(\ANDcolor)&*(\ANDcolor)\solid&\\
\none&&&&&&*(\ORcolor)&*(\ORcolor)&*(\ORcolor)&*(\ORcolor)&*(\ORcolor)&*(\ORcolor)\lrtri&*(\ORcolor)\solid&*(\ORcolorBack)&*(\ORcolorBack)&*(\ORcolorBack)&*(\ORcolorBack)&*(\ORcolorBack)&*(\ORcolorBack)&*(\ORcolorBack)&*(\ANDcolor)\solid&*(\ANDcolor)&*(\ANDcolor)\solid&&\\
\none&&&&&&*(\ORcolor)&*(\ORcolor)&*(\ORcolor)&*(\ORcolor)\lrtri&*(\ORcolor)&*(\ORcolor)&*(\ORcolor)\solid&*(\ORcolorBack)&*(\ORcolorBack)&*(\ORcolorBack)&*(\ORcolorBack)&*(\ORcolorBack)&*(\ORcolorBack)&*(\ORcolorBack)&*(\ANDcolor)\solid&*(\ANDcolor)&*(\ANDcolor)\solid&&\\
\none&&&&&&*(\ORcolor)&*(\ORcolor)\lrtri&*(\ORcolor)\solid&*(\ORcolor)\solid&*(\ORcolor)\solid&*(\ORcolor)\solid&*(\ORcolor)\solid&*(\ORcolorBack)&*(\ORcolorBack)&*(\ORcolorBack)&*(\ORcolorBack)&*(\ORcolorBack)&*(\ORcolorBack)&*(\ORcolorBack)&*(\ANDcolor)\solid&*(\ANDcolor)&*(\ANDcolor)\solid&&\\
\none&&&&&&*(\ORcolorBack)&*(\ORcolorBack)&*(\ORcolorBack)&*(\ORcolorBack)&*(\ORcolorBack)&*(\ORcolorBack)&*(\ORcolorBack)&*(\ORcolor)\solid&*(\ORcolor)&*(\ORcolor)\solid&*(\ORcolor)&*(\ORcolor)\solid&*(\ORcolor)&*(\ORcolor)\solid&*(\ANDcolor)\solid&*(\ANDcolor)&*(\ANDcolor)\solid&&\\
\none&&&&&&*(\ORcolorBack)&*(\ORcolorBack)&*(\ORcolorBack)&*(\ORcolorBack)&*(\ORcolorBack)&*(\ORcolorBack)&*(\ORcolorBack)&*(\ORcolor)\solid&*(\ORcolor)\ultri&*(\ORcolor)\solid&*(\ORcolor)\ultri&*(\ORcolor)\solid&*(\ORcolor)\ultri&*(\ORcolor)\solid&*(\ANDcolor)\solid&*(\ANDcolor)&*(\ANDcolor)\solid&&\\
\none&&&&&&*(\ORcolorBack)&*(\ORcolorBack)&*(\ORcolorBack)&*(\ORcolorBack)&*(\ORcolorBack)&*(\ORcolorBack)&*(\ORcolorBack)&*(\ORcolor)\solid&*(\ORcolor)&*(\ORcolor)&*(\ORcolor)&*(\ORcolor)&*(\ORcolor)&*(\ORcolor)&*(\ANDcolor)&*(\ANDcolor)\movable&*(\ANDcolor)&*(\ANDcolor)\solid&\\
\none&&&&&&*(\ORcolorBack)&*(\ORcolorBack)&*(\ORcolorBack)&*(\ORcolorBack)&*(\ORcolorBack)&*(\ORcolorBack)&*(\ORcolorBack)&*(\ORcolor)&*(\ORcolor)&*(\ORcolor)&*(\ORcolor)&*(\ORcolor)&*(\ORcolor)\lrtri&*(\ORcolor)\solid&*(\ANDcolor)\solid&*(\ANDcolor)&*(\ANDcolor)\solid&&\\
\none&&&&&&*(\ORcolorBack)&*(\ORcolorBack)&*(\ORcolorBack)&*(\ORcolorBack)&*(\ORcolorBack)&*(\ORcolorBack)&*(\ORcolorBack)&*(\ORcolor)&*(\ORcolor)&*(\ORcolor)&*(\ORcolor)\lrtri&*(\ORcolor)&*(\ORcolor)&*(\ORcolor)\solid&*(\ANDcolor)\solid&*(\ANDcolor)&*(\ANDcolor)\solid&&\\
\none&&&&&&*(\ORcolorBack)&*(\ORcolorBack)&*(\ORcolorBack)&*(\ORcolorBack)&*(\ORcolorBack)&*(\ORcolorBack)&*(\ORcolorBack)&*(\ORcolor)&*(\ORcolor)\lrtri&*(\ORcolor)\solid&*(\ORcolor)\solid&*(\ORcolor)\solid&*(\ORcolor)\solid&*(\ORcolor)\solid&*(\ANDcolor)\solid&*(\ANDcolor)&*(\ANDcolor)\solid&&\\
\none&&&&&&&&&&&&&&&&&&&&*(\ANDcolor)&*(\ANDcolor)\lrtri&*(\ANDcolor)\solid&&\\
\none&&&&&&&&&&&&&&&&&&&&*(\ANDcolor)\solid&*(\ANDcolor)\solid&*(\ANDcolor)\solid&&\\
\none&\solid&\solid&\solid&\solid&\solid&\solid&\solid&\solid&\solid&\solid&\solid&\solid&\solid&\solid&\solid&\solid&\solid&\solid&\solid&\solid&\solid&\solid&\solid&\solid\\
\end{ytableau}
$};
\end{tikzpicture}
\caption{The puzzle corresponding to the expression
$(A\vee B\vee \neg C)\wedge (A\vee \neg B\vee C)$.
Note that if $A=\mathtt{true}$, the expression is satisfied.
Thus this particular puzzle can be solved without deciding truth values for the variables $B$ and $C$,
and the movable blocks in the $B$ and $C$ variable gadgeds do not need to be moved.
}\label{fig:fullExample}
\end{figure}
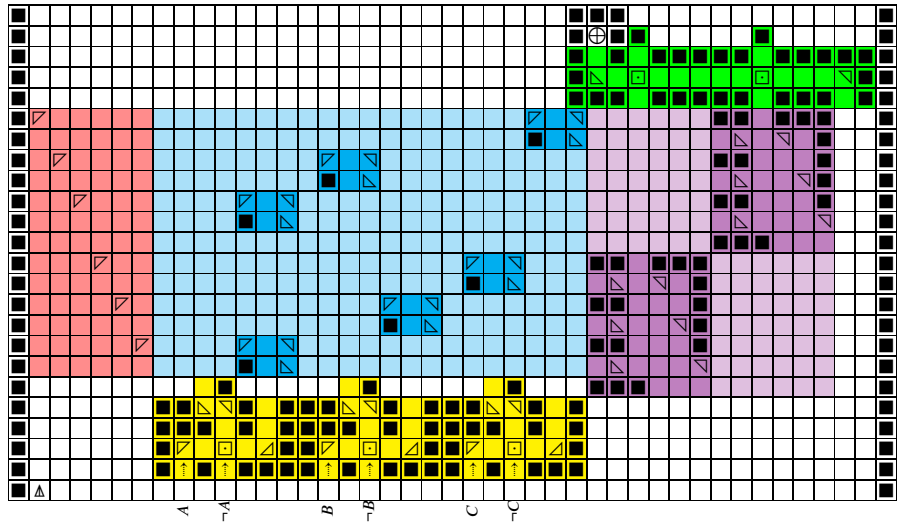

The following lemma shows that solving LaserTank puzzles can be done in polynomial time with 
a non-deterministic Turing machine. Hence LaserTank is in $\mathrm{NP}$.
\begin{lemma}\label{lem:PolyRunTIme}
A solution consisting of $k$ steps to a LaserTank puzzle on a board of size $n$ 
can be verified in time $O(kn)$.
\end{lemma}
\begin{proof}
It is straightforward to show that the laser movement is time-reversible.
This implies that it is impossible for a laser shot by the tank to end up in an ``infinite loop'' 
while being reflected by mirrors. Remember also that the laser stops as soon as it hits a solid block,
a movable block, or moves a mirror.
It follows that after firing the laser, it takes less than $4n$ steps before the laser finds its final 
destination, where $n$ is the number of tiles on the board.
Simulating a sequence of $k$ moves thus requires $O(kn)$ time.
\end{proof}

From our construction, it is a straightforward calculation 
to see that given a $3$-SAT expression with $V$ variables 
and $C$ clauses gives a puzzle contained on a 
board with size $(7V+9C+4)(7C+9)$.
This is evidently polynomial in the size of the expression.

\begin{proof}[of \cref{thm:MainTheorem}]
A $3$-SAT problem can be converted to a LaserTank puzzle in polynomial time 
since the board size is a polynomial in the number of variables and clauses.
Furthermore, a solution to such a LaserTank puzzle can easily be translated 
back to a solution of the original $3$-SAT problem in polynomial time,
by simply performing all the steps. Note that a LaserTank puzzle solution might not 
decide the truth value of some variables (see caption of \cref{fig:fullExample}), 
in which case, one may simply let these values be \texttt{true}.
According to \cref{lem:PolyRunTIme},
the translation of a puzzle solution to a $3$-SAT solution only requires a 
polynomial time in the input (number of steps).
This shows that LaserTank is at least as hard as $3$-SAT. 
Finally, \cref{lem:PolyRunTIme} shows that a solution can be verified in polynomial time
and hence LaserTank is NP-complete.
\end{proof}

Notice that in both the \texttt{and}-- and \texttt{literal}-gadget, each
movable block can be replaced with a $\ultri$-mirror without changing the behavior of the gadget. 
Thus \cref{thm:MainTheorem} is valid even in the case when restricting to puzzles without movable blocks.
Furthermore, we can extend \cref{thm:MainTheorem} to the case where the tank can turn and move in all four directions. 
To do so, we need to make sure the tank only has access to the same inputs as in the previous setup.
This can be done by inserting additional rows in the puzzle such that every other row is empty,
and then inserting two columns in with the pattern 
$\ytableausetup{boxsize=0.75em}
\begin{ytableau}
\urtri & \ultri \\ 
\lltri & \lrtri
\end{ytableau}$
between the initial position of the tank and the rest of the board.
We leave the details to the reader.

\bibliographystyle{splncs04}
\bibliography{bibliography}

\begin{thebibliography}{1}
\providecommand{\url}[1]{\texttt{#1}}
\providecommand{\urlprefix}{URL }
\providecommand{\doi}[1]{https://doi.org/#1}

\bibitem{Cook1971}
Cook, S.A.: The complexity of theorem-proving procedures. In: Proceedings of
  the Third Annual ACM Symposium on Theory of Computing. pp. 151--158. STOC
  '71, ACM, New York, NY, USA (1971). \doi{10.1145/800157.805047}

\bibitem{DorZwick1999}
Dor, D., Zwick, U.: {SOKOBAN} and other motion planning problems. Computational
  Geometry  \textbf{13}(4),  215--228 (Oct 1999).
  \doi{10.1016/s0925-7721(99)00017-6}

\bibitem{DemaineHohenbergerLibenNowell2002}
Erik D.~Demaine, S.H., Liben-Nowell, D.: Tetris is hard, even to approximate.
  Tech. Rep. MIT-LCS-TR-865, MIT, Cambridge (2002),
  \url{https://arxiv.org/abs/cs/0210020}

\bibitem{FlakeBaum2002}
Flake, G.W., Baum, E.B.: {R}ush {H}our is {PSPACE}-complete, or ``{W}hy you
  should generously tip parking lot attendants''. Theoretical Computer Science
  \textbf{270}(1-2),  895--911 (Jan 2002). \doi{10.1016/s0304-3975(01)00173-6}

\bibitem{Kaye2000}
Kaye, R.: Minesweeper is {NP}-complete. The Mathematical Intelligencer
  \textbf{22}(2),  9--15 (Mar 2000). \doi{10.1007/bf03025367}

\bibitem{Viglietta2013}
Viglietta, G.: Gaming is a hard job, but someone has to do it! Theory of
  Computing Systems  \textbf{54}(4),  595--621 (Aug 2013).
  \doi{10.1007/s00224-013-9497-5}

\end{thebibliography}

\end{document}